\documentclass[copyright,creativecommons]{eptcs}

\usepackage{amssymb,amsmath,amstext,latexsym,xfrac}
\usepackage{amsthm}
\usepackage{color}
\usepackage{graphicx}
\usepackage{url}

\providecommand{\sortnoop}[1]{}
\usepackage[noadjust]{cite}

\usepackage{subfigure,rotating}
\usepackage{picins}

\newcommand{\blankline}{\vspace{1.0\baselineskip}}
\newcommand{\halfblankline}{\vspace{0.5\baselineskip}}

\usepackage{tikz}
\usetikzlibrary{automata,arrows,shapes}
\usetikzlibrary{positioning}

\makeatletter
\newsavebox{\@brx}
\newcommand{\llangle}[1][]{\savebox{\@brx}{\(\m@th{#1\langle}\)}%
  \mathopen{\copy\@brx\mkern2mu\kern-0.9\wd\@brx\usebox{\@brx}}}
\newcommand{\rrangle}[1][]{\savebox{\@brx}{\(\m@th{#1\rangle}\)}%
  \mathclose{\copy\@brx\mkern2mu\kern-0.9\wd\@brx\usebox{\@brx}}}
\makeatother

\providecommand{\event}{FMSPLE 2016}
\providecommand{\firstpage}{1}
\providecommand{\sortnoop}[1]{}

\newcommand{\BExpF}{\mathbb{B}[ \mkern1mu \calF \mkern1mu ]}

\newcommand{\TRUE}{\textbf{true}}

\newcommand{\bbB}[1]{\lbb {#1} \rbb}
\newcommand{\bbF}[1]{\lbb {#1} \rbb_F}
\newcommand{\bbFp}[1]{\lbb {#1} \rbb_{F|p}}
\newcommand{\bbFF}[1]{\lbb {#1} \rbb_{F}^{\prime}}

\newcommand{\bbL}[1]{\lbb {#1} \rbb_L}

\newcommand{\bnfeq}{\mathrel{{:}{:}=}}
\newcommand{\bftwo}{\textbf{\textrm{2}}}
\newcommand{\bools}{\mathbb{B}}

\newcommand{\eps}{\varepsilon}

\newcommand{\falsum}{\mathord{\perp}}
\newcommand{\lbb}{\mathopen{\lbrack \! \lbrack \mkern2mu}}
\newcommand{\lc}{\mathopen{\lbrace \,}}
\newcommand{\modelsF}{\mathrel{\models_{\mkern-1mu F}}}
\newcommand{\modelsFprime}{\mathrel{\models'_{\mkern-1mu F}}}
\newcommand{\modelsFp}{\mathrel{\models_{\mkern-1mu F \mkern-1mu | \mkern0.5mu p}}}
\newcommand{\modelsL}{\mathrel{\models_{\mkern-1mu L}}}
\newcommand{\mCRL}{{\texttt{\MakeLowercase{m}CRL2}}}
\newcommand{\muL}{\mu \mkern-1mu L}
\newcommand{\muLFO}{\mu \mkern-1mu L \mkern-0.5mu {}_{\textrm{\tiny FO}}}
\newcommand{\muLf}{\mu \mkern-1mu L \mkern-1mu {}_f}

\newcommand{\muLfprime}{\mu \mkern-1mu L' \mkern-6mu {}_f}
\newcommand{\mybox}[1]{\lbrack {#1} \rbrack \mkern1mu}

\newcommand{\mydiamond}[1]{\langle {#1} \rangle}
\newcommand{\myifthenelse}[3]{\text{\textbf{if} ${#1}$ \textbf{then} ${#2}$ \textbf{else} ${#3}$ \textbf{end}}}
\newcommand{\myruby}[1]{\llangle {#1} \rrangle}
\newcommand{\ofsort}{\mkern-2mu \mathord{\colon} \mkern-6mu}
\newcommand{\powerset}[1]{\mathrm{\textbf{\rm 2}}^{#1}}
\newcommand{\powersetP}{\powerset{\mkern1mu \calP}}
\newcommand{\powersetS}{\powerset{\mkern1mu S}}
\newcommand{\powersetSP}{\powerset{\mkern1mu S \times \calP}}

\newcommand{\psif}{\psi_{\mkern-1mu f}}
\newcommand{\rbb}{\mathclose{\mkern2mu \rbrack \! \rbrack}}
\newcommand{\rc}{\mathclose{\, \rbrace}}

\newcommand{\singleton}[1]{\lbrace {#1} \rbrace}
\newcommand{\sinit}{s_{\mkern-1mu {\ast}}}
\newcommand{\sortD}{\ensuremath{\mathrm{D}}}

\newcommand{\trans}[2]{\xrightarrow{{#1}}_{#2}}
\newcommand{\transF}[2]{\xrightarrow{{#1}| \mkern0.5mu {#2}}_F}

\newcommand{\transFp}[1]{\xrightarrow{#1}_{F|p}}
\newcommand{\ttf}{\textsf{f}}
\newcommand{\ttg}{\textsf{g}}

\newcommand{\varphif}{\varphi_{\mkern-1mu \mathalpha{f}}}
\newcommand{\varv}{\mathrm{v}}
\newcommand{\verum}{\mathord{\top}}

\newcommand{\PSet}{\textsl{P$\mkern1mu$Set}}

\newcommand{\fmf}{\textsl{fm} \mkern1mu}
\newcommand{\fp}{\textsl{fp} \mkern1mu}
\newcommand{\gfp}{\textsl{gfp}}
\newcommand{\lfp}{\textsl{lfp}}

\newcommand{\sEnv}{\textsl{s$\mkern1mu$Env}}
\newcommand{\sPEnv}{\textsl{sPEnv}}
\newcommand{\spEnv}{\textsl{spEnv}}
\newcommand{\sSet}{\textsl{s$\mkern1mu$Set}}
\newcommand{\sPSet}{\textsl{sPSet}}
\newcommand{\spSet}{\textsl{spSet}}
\newcommand{\smf}{\textsl{sm} \mkern1mu}

\newcommand{\calA}{\mathcal{A}}

\newcommand{\calF}{\mathcal{F}}
\newcommand{\calP}{\mathcal{P}}

\newcommand{\calX}{\mathcal{X}}

\newtheorem{theorem}{Theorem}
\newtheorem{lemma}[theorem]{Lemma}
\newtheorem{definition}[theorem]{Definition}

\newtheorem{example}[theorem]{Example}

\title{Towards a Feature mu-Calculus Targeting SPL Verification}

\author{%
  Maurice H. ter Beek   
  \institute{ISTI--CNR, Pisa, Italy} 
  \email{m.terbeek@isti.cnr.it}
\and
  Erik P. de Vink
  \institute{TU/e, Eindhoven, The Netherlands}
  \institute{CWI, Amsterdam, The Netherlands}
  \email{e.p.d.vink@tue.nl}
\and
  Tim A. C. Willemse
  \institute{TU/e, Eindhoven, The Netherlands}
  \email{t.a.c.willemse@tue.nl}
}

\def\titlerunning{Towards a Feature mu-Calculus Targeting SPL Verification}
\def\authorrunning{M.H. ter Beek, E.P. de Vink \& T.A.C. Willemse}

\begin{document}

\maketitle

\begin{abstract}
  \textbf{Abstract} The modal $\mu$-calculus $\muL$ is a well-known
  fixpoint logic to express and model check properties interpreted
  over labeled transition systems.  In this paper, we propose two
  variants of the $\mu$-calculus, $\muLf$ and $\muLfprime$, for
  feature transition systems.  For this, we explicitly incorporate
  feature expressions into the logics, allowing operators to select
  transitions and behavior restricted to specific products and
  subfamilies. We provide semantics for $\muLf$ and $\muLfprime$ and
  relate the two new $\mu$-calculi and~$\muL$ to each other.  Next, we
  focus on the analysis of SPL behavior and show how our formalism can
  be applied for product-based verification with $\muLf$ as well as
  family-based verification with $\muLfprime$.  We illustrate by means
  of a toy example how properties can be model checked, exploiting an
  embedding of $\muLfprime$ into the $\mu$-calculus with data.
\end{abstract}


\hyphenation{sub-formula}


\section{Introduction}
\label{intro}

Formal methods and analysis tools for the specification and
verification of SPL models are widely
studied~\cite{SH11,BCLW13,TAKSS14}. Since many SPL
applications concern embedded, safety-critical systems, guaranteeing
their correct behavior by means of formal verification is an important
subject of study.  
However, when the
system to be analyzed is a product line, i.e.\ a family of systems,
the number of possible products leads to an exponential blow-up:
growing numbers of products on top of increasing numbers of
states. Hence, enumerative product-by-product analysis methods will
quickly prove infeasible for larger SPL models. Family-based
verification, as opposed to product-based verification, seeks to
exploit the commonalities of products that underlies a product line to
tackle large SPL models~\cite{TAKSS14}.

In recent years, numerous variants of known behavioral models have
been tailored to deal with the variability of SPL with the aim of
verifying temporal properties of SPL models. These include modal
transition systems (MTS)~\cite{FUB06,BFGM15b}, I/O
automata~\cite{LNW07,LPT09}, process
calculi~\cite{EW11,LT12,BLP13,LMBR14,Tri14} and feature
transition systems (FTS)~\cite{CCSHLR13,CCHLS14}.
In particular the latter have gained substantial popularity: they
offer a compact representation of a family of product behaviors,
individually modeled as labeled transition systems (LTS), in a single
transition system model in which actions are guarded by feature
expressions whose satisfaction (or not) indicate the presence (or
absence) of these actions in product behaviors.  This has resulted in
dedicated SPL model checkers~\cite{BMS12,CCHLS12,CCHSL13} as well as
the application of existing model checkers like
\texttt{N\MakeLowercase{u}SMV}~\cite{CCHLS14}, \mCRL~\cite{BDV14a} 
and \texttt{FMC}~\cite{BFGM15a} to SPL\@.

In~\cite{BDV14a,BDV14c,BBDV15}, we showed how the formal specification 
language \mCRL{} and its industrial-strength tool\-set can be
exploited to model and analyze SPL\@. The \mCRL{} tool\-set supports
parametrized modeling, model reduction and quality assurance
techniques like model checking. For more details, the reader is
referred to~\cite{CGKSVWW13,GM14} and \url{www.mcrl2.org}.  In
particular, we illustrated the use of \mCRL's parametrized data
language to model and select valid product configurations, in the
presence of feature attributes and quantitative constraints, and to
model and check the behavior of individually generated products (or of
a set of products, by tweaking the selection process).  Hence, the SPL
model-checking analyses with \mCRL{} studied so far fall in the
category of product-based analyses.  While we did equip our \mCRL{}
models of product families with an FTS-like semantics, to be able to
perform family-based verification also the supporting logic (a variant
of the first-order modal $\mu$-calculus augmented with data) needs to
be able to deal with the transitions of FTS labeled with feature
expressions.
 
The modal $\mu$-calculus $\muL$, going back to~\cite{Koz83}, is used
to express and model check properties interpreted over LTS, which
subsumes more intuitive popular temporal logics like LTL and CTL\@.
The model-checking approaches
of~\cite{CHSLR10,CCHLS12,CCHSL13,CCSHLR13,BLP13} are based
on~LTL, those of~\cite{LPT09,BMS12,CHSL11,CCHLS14,BFGM15b}
on~CTL and those of~\cite{LT12,BDV14a,BDV14c,BBDV15,BFGM15a,LMBR16} 
on the $\mu$-calculus.  In line with the recommendations from~\cite{ABDFS13}
to \lq\lq adopt and extend state-of-the-art analysis tools\rq\rq\ and
to \lq\lq analyze feature combinations corresponding to products of
the product line\rq\rq, in this paper we propose two variants of
$\muL$, coined $\muLf$ and $\muLfprime$, for FTS\@. For this, we
explicitly incorporate feature expressions into the logics, thus
allowing operators to single out transitions and behavior restricted
to specific products and subfamilies. We provide semantics for $\muLf$
and~$\muLfprime$ and relate the three logics to each other. 
Given that LTL and CTL are strict, partly overlapping subsets of $\muL$, 
each of the feature-oriented variants introduced in this paper can 
express properties that the approaches based on LTL or CTL cannot
(cf.~\cite{CGP99} for examples of such properties).

In line with the extensions feature LTL (fLTL)~\cite{CCSHLR13} and 
feature CTL (fCTL)~\cite{CCHLS14}, we extend the standard 
$\mu$-calculus to account for feature expressions that define the 
set of products over which a formula is to be verified. 
However, while fLTL and fCTL do not change the semantics of the 
temporal operators, but only limit or parametrize the set of products 
over which they are evaluated by the addition of a feature expression 
as quantifier or guard, we do change the semantics.
In detail, we replace operators $\mydiamond{a}$
and~$\mybox{a}$ of~$\muL$ by \lq{feature}\rq\ operators
$\mydiamond{a|\chi}$ and~$\mybox{a|\chi}$ for~$\muLf$ and
$\myruby{a|\chi}$ and~$\mybox{a|\chi}$ for~$\muLfprime$, for $\chi$ an
arbitrary feature expression. 
Intuitively, the classical diamond operator $\mydiamond{a} \mkern1mu \varphi$ 
(may modality) is valid if there exists an $a$-transition that leads to a state satisfying 
$\varphi$, while the classical box operator $[a] \mkern1mu \varphi$ 
(must modality) is valid if all $a$-transitions lead to a state where $\varphi$ is valid 
(i.e.\ if no such transition exists it holds trivially). 

The logic~$\muLf$ is product-oriented. Informally, a product~$p$ satisfies 
formula $\mydiamond{a|\chi} \mkern1mu \varphi$ with respect to an 
FTS~$F$ in some state, if $p$ meets~$\chi$ and an $a$-transition exists 
for~$p$ in~$F$ to a state where the formula~$\varphi$ holds
for~$p$. Similarly, $p$ satisfies $\mybox{a|\chi} \mkern1mu \varphi$
with respect to~$F$ in a state, if $p$ meets~$\chi$ and for all
$a$-transitions for~$p$ in~$F$ the formula~$\varphi$ holds for~$p$ in
the target state, or $p$~does not meet~$\chi$. So, $\mydiamond{a|\chi}
\mkern1mu \varphi$ \emph{does not} hold if $p$ does not satisfy the
feature expression~$\chi$, while $\mybox{a|\chi} \mkern1mu \varphi$
\emph{does} hold if $p$ does not satisfy 
$\chi$. 

The logic~$\muLfprime$ is family-oriented. The formula 
$\myruby{a|\chi} \mkern1mu \varphi$ holds for a set of products~$P$ 
with respect to FTS~$F$ in a state~$s$, if all products in~$P$ meet 
the feature expression~$\chi$ and there exists a single $a$-transition, 
possible for all products in~$P$, to a state where $\varphi$ holds 
for the set~$P$. In a way, the modality~$\myruby{a|\chi}$ of~$\muLfprime$ 
is a global variant of the local modality~$\mydiamond{a|\chi}$ of~$\muLf$. 
A formula $\mybox{a|\chi} \mkern1mu \varphi$ of~$\muLfprime$ holds 
in a state of~$F$ for a set of products~$P$, if for each subset~$P'$ 
of~$P$ for which an $a$-transition is possible, for all products of the
subset~$P'$ the formula~$\varphi$ holds for~$P'$ in the target
state. Note how, on the one hand, $\mybox{a|\chi} \mkern1mu \varphi$ 
is fulfilled for~$P$ if no product of~$P$ meets the feature expression~$\chi$. 
On the other hand, $\varphi$ is checked against subsets~$P'$ of~$P$, 
cut out by~$\chi$ and the feature expressions decorating the transitions 
of the FTS~$F$.

Jumping ahead, for a product~$p$ and a formula~$\varphi$ of~$\muLf$,
possibly involving feature expressions in the modalities, we have a
corresponding formula~$\varphi_{p}$ without feature expressions
in~$\muL$ and a corresponding formula~$\varphi'$ using
$\myruby{a|\chi}$ rather than~$\mydiamond{a|\chi}$ in~$\muLfprime$. 
In this paper, we will show
\begin{displaymath}
  \forall p \in P \colon {} \modelsFp \varphi_{p}
  \quad \stackrel{(i)}{\iff} \quad
  \forall p \in P \colon p \modelsF \varphi
  \quad \stackrel{(ii)}{\impliedby} \quad
  P \modelsFprime \varphi'
\end{displaymath}
for all sets of products~$P$ and given an FTS~$F$. Thus, (i)~$\varphi$
holds for~$p$ with respect to the FTS~$F$ iff $\varphi_p$ holds with
respect to the LTS~$F|p$, which is the projection of~$F$ on~$p$ 
obtained by including an $a$-transition in $F|p$ iff $p \in \gamma$ 
for a transition of~$F$ labelled with $a | \mkern0.5mu\gamma$. 
And~(ii), with respect to~$F$, $\varphi$~holds for all
products~$p$ in a family of products~$P$ if the formula~$\varphi'$
holds for the family~$P$. This provides us both with a correct
FTS-based semantics of~$\muLf$ with respect to a standard LTS-based
semantics, due to~(i), and with a possibility for family-based
model checking due to~(ii). However, the latter requires that we
have means to actually verify $\muLfprime$-formulas. We provide an
outline for this exploiting the \mCRL{} toolset. We sketch an
embedding of $\muLfprime$ into the $\mu$-calculus with data for a small
example.  We have already started to work on larger SPL models from
the literature, such as the well-known minepump model on which we 
evaluate our approach in a companion paper~\cite{BDVW16b}.

The paper is organized as follows. Section~\ref{epsilon} describes the
starting point of the research described in this paper: product-based
verification of product-based behavior with~$\muL$. In
Section~\ref{eta}, we introduce the $\muLf$-variant of~$\muL$ and
prove the soundness of its semantics, after which we introduce the
$\muLfprime$-variant in Section~\ref{zeta} and show how it can be used
for family-based verification of family-based behavior in
Section~\ref{sec-mc-zeta}. We then discuss our results and planned future
work in Section~\ref{conclusion}.


\section{Product-based behavior---product-based verification}
\label{epsilon}

In this section we recall the definition of an LTS and of a variant~$\muL$
of Kozen's modal $\mu$-calculus~\cite{Koz83}, and its semantics.

\halfblankline 

\begin{definition}
  \label{df-lts}
  An LTS $L$ over the set~$\calA$, the set of actions, is a triple $L
  = ( S ,\, {\rightarrow} \mkern1mu ,\, \sinit )$, with $S$ a finite
  set, the set of states, ${\rightarrow} \subseteq {S \times \calA
    \times S}$ the transition relation, and $\sinit \in S$ the initial
  state.
\end{definition}

\halfblankline

\noindent
The modal $\mu$-calculus involves modalities $\mydiamond{a}$
and~$\mybox{a}$, and fixpoint constructions. Its formulas are to be
interpreted over LTSs. See~\cite{BS01} for an overview.

\halfblankline

\begin{definition}
  Fix a set~$\calX$ of variables, ranged over by~$X$. The
  $\mu$-calculus~$\muL$ over $\calA$ and~$\calX$, containing
  formulas~$\varphi$, is given by
  \begin{displaymath}
    \varphi \bnfeq
    \begin{array}[t]{@{}l}
      \falsum \mid \verum \mid \neg \varphi \mid \varphi \lor \psi
      \mid \varphi \land \psi \mid 
      \mydiamond{a}\varphi \mid [a]\varphi \mid 
      X \mid \mu X . \varphi \mid \nu X . \varphi
    \end{array}
  \end{displaymath}
  where for $\mu X . \varphi$ and $\nu X . \varphi$ all free
  occurrences of~$X$ in~$\varphi$ are in the scope of an even number
  of negations.
\end{definition}

\halfblankline

\begin{example}
  Assuming $a, b, c \in \calA$, formulas of~$\muL$ include the following: 
\begin{itemize}
  \item [(i)] $\mydiamond{a} ( \, \mybox{b} \falsum \land
    \mydiamond{c} \verum \, )$ \lq\lq it is possible to do action~$a$
    after which action~$b$ is not possible but action~$c$ is\rq\rq
\item [(ii)] $\mu X . \mkern1mu ( \mydiamond{a} X \lor \mydiamond{b}
  \verum )$ \lq\lq there exists a finite sequence of $a$-actions
  followed by a $b$-action\rq\rq
\item [(iii)] $\nu X . \mkern1mu \bigl( \mu Y . \mkern1mu \mybox{a} Y
  \land \mybox{b} X \bigr)$ \lq\lq on all infinite runs with actions
  $a$ and~$b$, action~$b$ occurs infinitely often\rq\rq
\end{itemize}
\end{example}

\halfblankline

\noindent
The syntactic restriction for a variable~$X$ that is bound by a
fixpoint construct to occur in the scope of an even number of
negations guarantees monotonicity, with respect to set inclusion, of
functionals used for the semantic definition below. From the
monotonicity it follows by the Knaster-Tarski theorem that the least
fixpoint and greatest fixpoint of the functionals exist.

\halfblankline

\begin{definition}
  \label{df-eps-semantics}
  Let an LTS~$L$, with set of states~$S$, be given. Define $\sSet$,
  the set of state sets by $\sSet = \powersetS$, and define $\sEnv$,
  the set of state-based environments, by $\sEnv = \calX \to
  \sSet$. Then the semantic function $\bbL{\cdot} : \muL \to \sEnv \to
  \sSet$ is given by
  \begin{align*}
      \bbL{ \falsum }(\eps) & = \varnothing \\
      \bbL{ \verum }(\eps) & = S \\
      \bbL{ \neg \varphi }(\eps) & =
      S \setminus \bbL{ \varphi }(\eps) 
      \displaybreak[2] \\
      \bbL{ ( \varphi \lor \psi ) }(\eps) & =
      \bbL{ \varphi }(\eps)
      \cup
      \bbL{ \psi }(\eps)
      \\
      \bbL{ ( \varphi \land \psi ) }(\eps) & =
      \bbL{ \varphi }(\eps)
      \cap
      \bbL{ \psi }(\eps)
      \smallskip \displaybreak[3] \\
      \bbL{ \mydiamond{a} \varphi }(\eps) & =
      \lc s \mid 
      \exists \mkern2mu t \colon s \trans{a}{} t \land
      t \mathop{\in} \bbL{\varphi}(\eps) 
      \rc
      \\
      \bbL{ \mybox{a} \varphi }(\eps) & =
      \lc s \mid 
      \forall t \colon s \trans{a}{} t \Rightarrow 
      t \mathop{\in} \bbL{\varphi}(\eps)
      \rc
      \smallskip \displaybreak[3] \\
      \bbL{ X }(\eps) & =
      \eps(X) 
      \\
      \bbL{ \mu X . \varphi }(\eps) & =
      \lfp( \, U \mapsto \bbL{ \varphi }(\eps[U/X]) \, )
      \\
      \bbL{ \nu X . \varphi }(\eps) & =
      \gfp( \, U \mapsto \bbL{ \varphi }(\eps[U/X]) \, )
  \end{align*}
\end{definition}

\noindent
By~$\eps[U/X]$ we denote the environment in~$\sEnv$ which
yields~$\eps(Y)$ for variables~$Y$ different from~$X$ and which yields
the set~$U$ for the variable~$X$. For an LTS~$L$ with initial
state~$\sinit$ and a closed $\muL$-formula~$\varphi$, we write
$\modelsL \mkern-1mu \varphi$ iff $\sinit \in \bbL{\varphi}(\eps_0)$,
where the environment $\eps_0 \in \sEnv$ is such that $\eps_0(X) =
\varnothing$ for all $X \in \calX$.



\halfblankline

\noindent
There are various approaches to model checking of $\muL$~formulas on
an LTS, in particular exploiting BDDs
and using equation
systems or parity games 
(see~\cite{CGP99,BS01} for more details and references). The
papers~\cite{BDV14a,BDV14c,BBDV15} present a model-checking 
approach to SPL using~$\muL$ for a prototypical coffee
machine. Basically, the approach in these papers is product-based, in
the sense that it provides a non-deterministic choice of the product
space before entering product behavior captured by an LTS\@. Thus,
although some tricks are possible, behavior is specified product-based
as is the verification approach.


\section{Family-based behavior---product-based verification}
\label{eta}

In this section we extend the logic~$\muL$ to a logic~$\muLf$
incorporating feature expressions in its modalities. Formulas
of~$\muLf$ are to be evaluated against an FTS as first proposed 
in~\cite{CHSLR10}. We note that, for the ease of presentation in 
this paper, the definition of an FTS below is slightly more abstract.

We fix a finite non-empty set~$\calF$ of features and a set~$\calA$ of
actions, and we let~$\BExpF$ denote the set of Boolean expressions
over~$\calF$. We have~$\textit{\ttf\/}$ as a typical element
of~$\calF$. Elements~$\chi$ and~$\gamma$ of~$\BExpF$ are referred to
as feature expressions. The constant~$\verum$ is used to denote the
feature expression that is always true. A product is a set of
features; $\calP$ is the set of products, thus $\calP \subseteq
\bftwo^{\calF}$. We use~$p$ to range over~$\calP$. A product~$p$
induces an assignment~$\alpha$ of features, viz.\ $\alpha_p : \calF
\to \bools$ with $\alpha_p(\textit{\ttf\/}) = \TRUE$ iff
$\textit{\ttf\/} \in p$. We write $p \in \chi$ for $\alpha_p \models
\chi$. We also identify a feature expression~$\chi$ with the set of
products that satisfy~$\chi$.

\halfblankline

\begin{definition}
  An FTS~$F$ over $\calA$ and~$\calF$ is a triple $F = ( S ,\, \theta
  ,\, \sinit )$, with $S$ a finite set, the set of states,
  $\theta : S \times \calA \times S \to \BExpF$ the transition
  constraint function, and~$\sinit \in S$ the initial state.
\end{definition}

\halfblankline

\noindent
For states~$s, t \in S$, we write $s \transF{a}{\chi} t$ if
$\theta(s,a,t) = \chi$.
Given an FTS $F = ( S ,\, \theta ,\,
\sinit )$ and a product~$p \in \calP$, the projection of~$F$ with
respect to~$p$ is the LTS $F|p = ( S ,\, {\rightarrow}_{F|p} ,\,
\sinit)$ over $\calA$ with $s \trans{a}{F|p} t$ iff $p \in \gamma$ for
a transition $s \transF{a}{\gamma} t$ of~$F$.

\halfblankline

\begin{example}
  \label{ex-coffee-machine}
  Consider the FTS $F$ modeling a family of coffee machines, 
  an SPL of four products, with the set of features $\{C,D,E\}$ 
  representing the presence of a \emph{clean/descale} unit, 
  a \emph{dollar} unit and a \emph{euro} unit, depicted below (left). 
  The actions of $F$ represent inserting a coin (action~$\mathit{ins}$), 
  pouring a standard or large coffee (actions~$\mathit{sd}$ and 
  $\ell \mkern-2mu g$) and cleaning/descaling the machine 
  (action~$\mathit{cd}$). A coffee machine accepts either dollar coins 
  or euro coins. Large cups of coffee require two coins and are only 
  available for dollar machines. Cleaning can only occur when the 
  machine is idle with no coins inserted. For the products 
  $p_{1} = \singleton{C,D}$ and $p_{2} = \singleton{E}$, we have the 
  projections $F|p_1$ and~$F|p_2$, also depicted below (right).
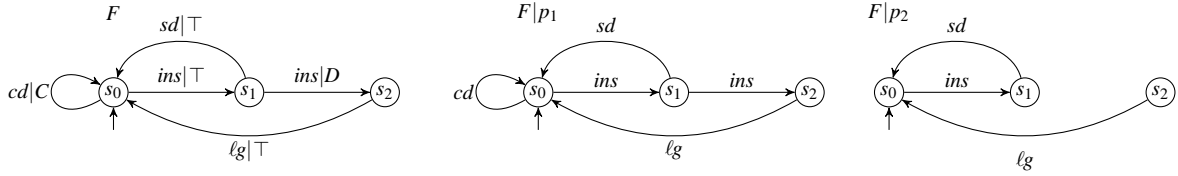
\begin{center}
  \begin{tikzpicture}[>=stealth', every state/.style={draw, minimum 
      size=10pt, inner sep=1pt},node distance=40pt] 
    \node[state] (r0) {\scriptsize $s_0$};
    \node[state, right=of r0] (r1) {\scriptsize $s_1$};
    \node[state, right=of r1] (r2) {\scriptsize $s_2$};
    \path ([xshift=-0pt,yshift=30pt] r0) node {\scriptsize $F$}; 

    \draw[->] 
    (r0)++(0,-0.5) -- (r0)
    (r0) edge node[yshift=-2pt,above] {\scriptsize $\mathit{ins} |
      \mkern1mu \verum$} (r1) 
    (r1) edge[out=105,in=75] node[yshift=-2pt,above] {\scriptsize
      $\mathit{sd} | \mkern1mu \verum$} (r0)
    (r1) edge node[yshift=-2pt,above] {\scriptsize $\mathit{ins} |
      \mkern0mu D$} (r2) 
    (r2) edge[bend left] node[yshift=2pt,below] {\scriptsize ${\ell}
      \mkern-1.5mu {g} | \mkern1mu \verum$} (r0)
    (r0) edge[loop,min distance=10mm,out=210,in=150]
    node[left,xshift=3pt] {\scriptsize $\mathit{cd} | \mkern1mu C$} (r0);

    \node[state, right of=r2, xshift=+18pt] (s0) {\scriptsize $s_0$};
    \node[state, right=of s0] (s1) {\scriptsize $s_1$};
    \node[state, right=of s1] (s2) {\scriptsize $s_2$};
    \path ([xshift=-0pt,yshift=30pt] s0) node {\scriptsize $F|p_{1}$}; 

    \draw[->] 
    (s0)++(0,-0.5) -- (s0)
    (s0) edge node[yshift=-2pt,above] {\scriptsize $\mathit{ins}$} (s1) 
    (s1) edge[out=105,in=75] node[yshift=-0.5pt,above] {\scriptsize
      $\mathit{sd}$} (s0)
    (s1) edge node[yshift=-2pt,above] {\scriptsize $\mathit{ins}$} (s2) 
    (s2) edge[bend left] node[yshift=2pt,below] {\scriptsize ${\ell}
      \mkern-1.5mu {g}$} (s0)
    (s0) edge[loop,min distance=10mm,out=210,in=150]
    node[left,xshift=3pt] {\scriptsize $\mathit{cd}$} (s0);

    \node[state, right of=s2, xshift=-10pt] (t0) {\scriptsize $s_0$};
    \node[state, right=of t0] (t1) {\scriptsize $s_1$};
    \node[state,right =of t1] (t2) {\scriptsize $s_2$};
    \path ([xshift=-0pt,yshift=30pt] t0) node {\scriptsize $F|p_{2}$}; 

    \draw[->] 
    (t0)++(0,-0.5) -- (t0)
    (t0) edge node[yshift=-2pt,above] {\scriptsize $\mathit{ins}$} (t1)
    (t1) edge[out=105,in=75] node[yshift=-0.5pt,above] {\scriptsize
      $\mathit{sd}$} (t0) 
    (t2) edge[bend left] node[yshift=-2pt,below] {\scriptsize $\ell
      \mkern-1.5mu {g}$} (t0);

  \end{tikzpicture}
\end{center}
  The LTS~$F|p_1$ has all transitions of~$F$, with the feature
  expressions decorating the arrows omitted. For~$F|p_2$, the
  $\mathit{cd}$-loop that depends on the feature~$C$ and the second
  insert action~$\mathit{ins}$ of the transition from $s_1$ to~$s_2$,
  which depends on the feature~$D$, are not present. Although~$s_2$ 
  is unreachable in~$F|p_2$ according to the definition, since 
  $p_2 \in \verum$, the transition from $s_2$ to~$s_0$ is present.
\end{example}

\halfblankline

\noindent
We next introduce the feature $\mu$-calculus~$\muLf$.

\halfblankline

\begin{definition}
  \label{df-muLf}
  The feature $\mu$-calculus~$\muLf$ over $\calA$, $\calF \!$, and~$\calX$,
  containing formulas~$\varphif$, is given by
  \begin{displaymath}
    \varphif \bnfeq
    \begin{array}[t]{@{}l}
      \falsum \mid \verum \mid \neg \varphif \mid \varphif \lor \psif
      \mid \varphif \land \psif \mid 
      \mydiamond{a|\chi} \mkern1mu \varphif \mid 
      [a|\chi] \mkern1mu \varphif \mid 
      X \mid \mu X . \mkern1mu \varphif \mid 
      \nu X . \mkern1mu \varphif
    \end{array}
  \end{displaymath}
  where for $\mu X . \mkern1mu \varphif$ and $\nu X . \mkern1mu
  \varphif$ all free occurrences of~$X$ in~$\varphif$ are in the scope
  of an even number of negations.
\end{definition}

\halfblankline

\begin{example}
  Assuming feature expressions $\chi, \chi_1, \chi_2 \in \BExpF$,
  the following are example formulas of~$\muLf$:
  \begin{itemize}
  \item [(i)] $\mydiamond{a | \chi_1 \land \chi_2} ( \, \mybox{b |
    \chi_1} \falsum \land \mydiamond{c | \chi_2} \verum \, )$ \lq\lq it is
    possible to do action~$a$ for products satisfying $\chi_1 \land
    \chi_2$ after which action~$b$ is not possible for products
    satisfying~$\chi_1$ but action~$c$ is for products
    satisfying~$\chi_2$\rq\rq
  \item [(ii)] $\mu X . \mkern1mu ( \mydiamond{a|\chi_1} X \lor
    \mydiamond{b|\chi_2} \verum )$ \lq\lq there exists a finite
    sequence of $a$-actions for products satisfying~$\chi_1$, followed by
    a $b$-action for products satisfying~$\chi_2$ as well\rq\rq
  \item [(iii)] $\nu X . \mkern1mu \bigl( \mu Y . \mkern1mu \mybox{a
    | \chi } Y \land \mybox{b | \mkern1mu \verum } X \bigr)$ \lq\lq for
    products satisfying~$\chi$ infinitely often action~$b$ can occur,
    after a finite number of times action~$a$\rq\rq
  \end{itemize}
\end{example}

\halfblankline

\noindent
The semantic function $\bbF{\cdot}$ given below returns for an FTS~$F$
all state-product pairs in which a formula $\varphif \in \muLf$ holds.

\halfblankline

\begin{definition}
  \label{df-eta-semantics}
  Let an FTS~$F$, with set of states~$S$, be given. Define $\spSet
  \mkern0.5mu$, the set of state-product sets, by $\spSet =
  \powersetSP$, and define $\spEnv$, the set of state-product
  environments, by $\spEnv = \calX \to \spSet$. Then the semantic
  function $\bbF{\cdot} : \muLf \to \spEnv \to \spSet$ is given by
      \begin{align*}
      \bbF{ \falsum }(\eta) & = \varnothing \\
      \bbF{ \verum }(\eta) & = S \times \calP \\
      \bbF{ \neg \varphif }(\eta) & =
      (S \times \calP) \setminus \bbF{ \varphif }(\eta) 
      \displaybreak[2] \\
      \bbF{ ( \varphif \lor \psif ) }(\eta) & =
      \bbF{ \varphif }(\eta)
      \cup
      \bbF{ \psif }(\eta)
      \\
      \bbF{ ( \varphif \land \psif ) }(\eta) & =
      \bbF{ \varphif }(\eta)
      \cap
      \bbF{ \psif }(\eta)
      \smallskip \displaybreak[3] \\
      \bbF{ \mydiamond{a|\chi} \varphif }(\eta) & =
      \lc (s,p) \mid {p \mathop{\in} \chi} \land
      {} \\ & \qquad\ \qquad
      \bigl(
      \exists \mkern1mu \gamma, t \colon s \transF{a}{\gamma} t \land 
      p \mathop{\in} \gamma \land
      {(t,p) \mathop{\in} \bbF{\varphif}(\eta)}
      \bigr) \rc
      \\
      \bbF{ \mybox{a|\chi} \varphif }(\eta) & =
      \lc (s,p) \mid {p \mathop{\in} \chi} \Rightarrow
      {} \\ & \qquad\ \qquad
      \bigl(
      \forall \mkern1mu \gamma, t \colon s \transF{a}{\gamma} t \land 
      {p \mathop{\in} \gamma} \Rightarrow  
      {(t,p) \mathop{\in} \bbF{ \varphif}(\eta) }
      \bigr) \rc
      \smallskip \displaybreak[3] \\
      \bbF{ X }(\eta) & =
      \eta(X) 
      \\
      \bbF{ \mu X . \varphif }(\eta) & =
      \lfp( \, V \mapsto \bbF{ \varphif }(\eta[V/X]) \, )
      \\
      \bbF{ \nu X . \varphif }(\eta) & =
      \gfp( \, V \mapsto \bbF{ \varphif }(\eta[V/X]) \, )
      \end{align*}
\end{definition}


\noindent
For an FTS~$F$ with initial state~$\sinit$, a product $p \in \calP$,
and a closed $\muLf$-formula~$\varphif$, we write $p \modelsF \varphif$
iff $(\sinit \mkern1mu ,p) \in \bbF{\varphif}(\eta_0)$, where the
environment $\eta_0 \in \spEnv$ is such that $\eta_0(X) = \varnothing$
for all $X \in \calX$.

\blankline

\noindent
As we shall see, model checking of a $\muLf$-formula for an individual
product reduces to model checking a $\muL$-formula. To make this
precise we introduce the translation function $\smf : \muLf \times
\calP \to \muL$ as follows:
\begin{displaymath}
  \def\arraystretch{1.2}
  \begin{array}{r@{\:}c@{\:}lcr@{\:}c@{\:}l}
    \smf( \falsum , p ) & = & \falsum & \qquad &
    \smf( X , p ) & = & X \\
    \smf( \verum , p ) & = & \verum &&
    \smf( \mu X . \varphif , p ) & = & \mu X . \smf( \varphif , p ) \\
    \smf( \varphif \lor \psif , p ) & = & 
      \smf( \varphif) \lor \smf(\psif) &&
      \smf( \nu X . \varphif , p ) & = & \nu X . \smf( \varphif , p ) \\
    \smf( \varphif \land \psif , p ) & = & 
      \smf( \varphif) \land \smf(\psif) \\
    \multicolumn{7}{c}{%
      \smf( \mydiamond{a|\chi} \mkern1mu \varphif , p )
      \, = \,
      \myifthenelse{p \in \chi}{\mydiamond{a} \mkern1mu \smf( \varphif ,
       p )}{\falsum} 
   } \\
    \multicolumn{7}{c}{%
      \smf( \mybox{a|\chi} \mkern1mu \varphif , p )
      \, = \,
      \myifthenelse{p \in \chi}{\mybox{a} \mkern1mu \smf( \varphif ,  p
        )}{\verum} 
    } \\
  \end{array}
  \def\arraystretch{1.0}
\end{displaymath}
Thus, given a formula~$\varphif \in \muLf$, $\smf(\varphif)$ is the
$\muL$-formula obtained from~$\varphif$ by replacing a
subformula~$\mydiamond{a|\chi} \mkern1mu \psif$ by~$\falsum$ and a
subformula~$\mybox{a|\chi} \mkern1mu \psif$ by~$\verum$, respectively,
in case $p \notin \chi$, while omitting the feature expression~$\chi$
otherwise.

The semantic function $\bbFp{\cdot}$ in the theorem below is an
instance of Definition~\ref{df-eps-semantics}, returning for the
LTS~$F|p$ all states in which a formula $\varphi \in \muLf$ holds,
given an environment.

\halfblankline

\begin{theorem}
  \label{theorem-eta-semantics-vs-eps-semantics}
  Let a state $s \in S$ and a product $p \in \calP$ be
  given. Suppose the environments $\eta \in \spEnv$ and~$\eps \in
  \sEnv$ are such that $(s,p) \in \eta(X) \iff s \in
  \eps(X)$, for all $X \in \calX$. Then it holds that 
  \begin{displaymath}
    (s,p) \in \bbF{\varphif} (\eta) 
    \quad \text{iff} \quad 
    s \in \bbFp{\smf(\varphif) } (\eps)
  \end{displaymath}
  for all $\varphif \in \muLf$.
\end{theorem}

\begin{proof}
  The theorem is proven by induction on the structure
  of~$\varphif$. Here we exhibit two cases.

  Case~1, $\mydiamond{a|\chi}$: Suppose $p \in \chi$. Then $\smf(
  \mydiamond{a|\chi} \mkern1mu \varphif , p ) = \mydiamond{a} \mkern1mu
  \varphi$, with $\smf( \varphif ) = \varphi$. We have
  \begin{displaymath}
    \begin{array}{rcl}
      \multicolumn{3}{l}{(s,p) \in \bbF{\mydiamond{a|\chi}
          \mkern1mu \varphif}(\eta)} 
      \\ & \iff &
      {p \mathop{\in} \chi} \land
      \bigl(
      \exists \mkern1mu \gamma, t \colon
      s \transF{a}{\gamma} t \land
      {p \mathop{\in} \gamma} \land
      {(t,p) \mathop{\in} \bbF{\varphif}(\eta)}
      \bigr)
      \\ && \qquad
      \text{(definition $\bbF{ \cdot }$)}
      \\ & \iff &
      \exists \mkern1mu t \colon
      s \transFp{a} t \land
      {t \mathop{\in} \bbFp{\varphi}(\eps)}
      \\ & & \qquad
      \text{($p \in \chi$, definition $F|p$, and induction hypothesis)}
      \\ & \iff &
      s \in \bbFp{\mydiamond{a|\chi} \varphi}(\eps)
      \\ && \qquad
      \text{(definition $\bbFp{ \cdot }$)}
    \end{array}
  \end{displaymath}
  Suppose $p \notin \chi$. Then $\smf( \mydiamond{a|\chi} \mkern1mu
  \varphif , p ) = \falsum$. We have $(s,p) \notin \bbF{
    \mydiamond{a|\chi} \mkern1mu \varphif }(\eta)$ by definition
  of~$\bbF{ \cdot }$ and $s \notin \bbFp{ \falsum }(\eps)$ by
  definition of~$\bbFp{ \cdot }$.

  Case~2, $\mu X . \varphif$: Then $\smf( \mu X . \varphif ) = \mu X
  . \varphi$, with $\smf( \varphif ) = \varphi$. If suffices to prove
  \begin{gather}
    (s,p) \in \bigcup_{i=0}^\infty \: V_i
    \quad \text{iff} \quad
    s \in \bigcup_{i=0}^\infty \: U_i
    \label{eq-sp-in-Vi-iff-s-in-Ui}
  \end{gather}
  where $V_0 = \varnothing$, $V_{i{+}1} = \bbF{\varphif}(\eta[V_i/X])$
  and $U_0 = \varnothing$, $U_{i{+}1} =
  \bbFp{\varphi}(\eps[U_i/X])$. Define $\eta_i$ and~$\eps_i$, for $i
  \geqslant 0$, by $\eta_0 = \eta$, $\eta_{i{+}1} = \eta[V_i/X]$ and
  $\eps_0 = \eps$, $\eps_{i{+}1} = \eps[U_i/X]$. We claim
  \begin{displaymath}
    \forall \mkern1mu Y \in \calX \colon (s,p) \in \eta_i(Y)
    \iff
    s \in \eps_i(Y)
    \quad \text{and} \quad
    (s,p) \in V_i \iff s \in U_i
  \end{displaymath}
  \emph{Proof of the claim.} Induction on~$i$. Basis, $i=0$: Clear, by
  definition of $V_0$ and~$U_0$ and by the assumption on $\eta$
  and~$\eps$. Induction step, $i{+}1$: We first check $(s,p) \in
  \eta_{i{+}1}(Y) \iff s \in \eps_{i{+}1}(Y)$, only
  for~$Y=X$.
  \begin{displaymath}
    \def\arraystretch{1.2}
   \begin{array}{l}
      {(s,p) \in \eta_{i{+}1}(X)}
      \iff 
      (s,p) \in \eta[V_i/X](X)
      \iff
      (s,p) \in V_i
      \iff
      {} \\ \qquad \qquad
      s \in U_i
      \  \text{(by induction hypothesis for~$i$)}
      \iff 
      s \in \eps[U_i/X](X)
      \iff
      s \in \eps_{i{+}1}(X)
    \end{array}
    \def\arraystretch{1.0}
  \end{displaymath}
  Next we verify $(s,p) \in V_{i{+}1} \iff s \in
  U_{i{+}1}$.
  \begin{displaymath}
    \def\arraystretch{1.2}
    \begin{array}{l}
      {(s,p) \in V_{i{+}1}} 
      \iff
      (s,p) \in \bbF{\varphif}(\eta[V_i/X])
      \iff
      (s,p) \in \bbF{\varphif}(\eta_{i{+}1})
      \iff
      {} \\ \qquad \qquad
      s \in \bbFp{\varphi}(\eps_{i{+}1})
      \  \text{(by induction hypothesis for $\varphif$)}
      \iff
      {} \\ \qquad \qquad \qquad \qquad 
      s \in \bbFp{\varphi}(\eps[U_i/X])
      \iff
      s \in U_{i{+}1}
    \end{array}
    \def\arraystretch{1.0}
  \end{displaymath}
  From the claim Equation~(\ref{eq-sp-in-Vi-iff-s-in-Ui}) follows
  directly. 
\end{proof}

\halfblankline

\noindent
From Theorem~\ref{theorem-eta-semantics-vs-eps-semantics} we have the
following immediate consequence:
\begin{equation}
  \modelsFp \smf( \varphif , p )
  \iff
  p \modelsF \varphif
  \label{eq-eta-semantics-vs-eps-semantics}
\end{equation}
for all $p \in \calP$, $\varphif \in \muLf$ closed.

\blankline

\noindent
Theorem~\ref{theorem-eta-semantics-vs-eps-semantics} and its
corollary show the strong connection of LTS model checking
for~$\muL$ and FTS model checking for~$\muLf$. In the next section we
introduce an adapted feature $\mu$-calculus, called~$\muLfprime$, for
which there is a looser relationship of FTS model checking for~$\muLf$
and FTS model checking for~$\muLfprime$ which only is sound for
formulas without negation. However, because of the duality of~$\muLf$
regarding modalities and fixpoints
shown in the sequel of this section, the latter is not a major
restriction.

We note that although the behavior of products is specified from the
family perspective by means of an FTS, model checking based on
Definition~\ref{df-eta-semantics} will have single products as unit of
granularity. However, one may take advantage of compact
representations when building up pieces of information during the
recursive exploration of subformulas and groups of products. In
particular, as in~\cite{CHSL11,CCHLS14}, BDDs can be used to
efficiently represent subsets of~$S \times \calP$.

\halfblankline

\begin{lemma}
  \label{lemma-duality}
  It holds that 
  \begin{itemize}
  \item [(a)] $\bbF{ \neg \mybox{a|\chi} \varphi }(\eta) = \bbF{
    \mydiamond{a|\chi} \mkern1mu \neg \varphi }(\eta)$ \, and \, $\bbF{ \neg
    \mydiamond{a|\chi} \varphi }(\eta) = \bbF{ \mybox{a|\chi} \neg
    \varphi }(\eta)$, \ for all $a \in \calA$, $\chi \in \BExpF$,
    $\varphi \in \muLf$, $\eta \in \spEnv$, and
  \item [(b)] $\bbF{ \neg \mu X . \varphi }(\eta) = \bbF{ \nu X .
    \varphi \mkern1mu [ \neg X / X ] \mkern1mu }(\eta)$ \, and \, $\bbF{
    \neg \nu X . \varphi }(\eta) = \bbF{ \mu X . \varphi \mkern1mu [
      \neg X / X ] \mkern1mu }(\eta)$, \ for all ${X \in \calX}$, $\varphi
    \in \muLf$ and $\eta \in \spEnv$.
  \end{itemize}
\end{lemma}

\begin{proof}
  We only present part~(a), part~(b) being exactly as standard. We
  have
  \smallskip \\
  \begin{displaymath}
    \def\arraystretch{1.2}
    \begin{array}[b]{rcl}
      \multicolumn{3}{l}{\bbF{ \neg \mybox {a|\chi} \varphi }(\eta)}  
      \\ &\!\!=\!\!\!& 
      (S \times\!\calP) {\setminus}
      \lc (s,p) \mid {p \mathop{\in} \chi} \Rightarrow\!
      \bigl( \forall \mkern1mu \gamma, t \colon
      s \transF{a}{\gamma} t \land
      {p \mathop{\in} \gamma} \Rightarrow\!
      {(t,p)\!\mathop{\in} \bbF{ \varphi }(\eta)}
      \bigr) \rc
      \qquad \text{(by\,definition\,$\bbF{ \cdot }$)}
      \\[0.15em] &\!\!=\!\!\!&
      \lc (s,p) \mid 
      {p \mathop{\in} \chi} \land
      \bigl( 
      \exists \mkern1mu \gamma, t \colon s \transF{a}{\gamma} t \land
      {p \mathop{\in} \gamma} \land
      {(t,p)\!\mathop{\notin} \bbF{ \varphi }(\eta)}
      \bigr) \rc
      \\[0.35em] &\!\!=\!\!\!&
      \bbF{ \mydiamond{a|\chi} \mkern1mu \neg \varphi }(\eta)
      \qquad
      \text{(by definition $\bbF{ \cdot }$)}
    \end{array}
    \def\arraystretch{1.0}
  \end{displaymath}
  $\ \ \bbF{ \neg \mydiamond {a|\chi} \varphi}(\eta)$
  \vspace*{-.85\baselineskip}
    \begin{align*}
      &=\,\ 
      (S \times \calP) {\setminus}
      \lc (s,p) \mid {p \mathop{\in} \chi} \land
      \bigl( \exists \mkern1mu \gamma, t \colon
      s \transF{a}{\gamma} t \land
      {p \mathop{\in} \gamma} \Rightarrow\!
      {(t,p)\!\mathop{\in} \bbF{ \varphi }(\eta)}
      \bigr) \rc 
      \qquad \text{(by\,definition\,$\bbF{ \cdot }$)}
      \\ &=\,\ 
      \lc (s,p) \mid 
      {p \mathop{\notin} \chi} \lor
      \bigl( 
      \forall \mkern1mu \gamma, t \colon s \transF{a}{\gamma} t \land
      {p \mathop{\in} \gamma} \land
      {(t,p)\!\mathop{\notin} \bbF{ \varphi }(\eta)}
      \bigr) \rc
      \\[0.35em] &=\,\ 
      \bbF{ [a|\chi] \mkern1mu \neg \varphi }(\eta)
      \qquad
      \text{(by definition $\bbF{ \cdot }$)}
    \qedhere\end{align*}
\end{proof}

\halfblankline

\noindent
Recall that in~$\varphif$, all free occurrences of variables used 
in fixpoints are in the scope of an even number of negations.
The duality of the modalities and of the fixpoint operators,
together with the De Morgan laws, allow for (closed) formulas to
\lq{push}\rq\ negation inside without effecting the meaning of a
formula. Therefore, in~$\muLf$, every formula~$\varphif$ has a
negation-free equivalent formula~$\psif$.


\section{Family-based behavior---family-based verification}
\label{zeta}

In this section we consider a variation of the $\mu$-calculus~$\muLf$
introduced in Definition~\ref{df-muLf} that allows for model checking
taking sets of products as point of view.
Modalities~$\mydiamond{a|\chi}$ of~$\muLf$ will be replaced by
modalities~$\myruby{a|\chi}$. The point is that the transition
required in the semantics of~$\myruby{a|\chi}$ in the adapted
$\mu$-calculus requires the existence of a specific transition in the
underlying FTS that applies for all products under consideration at
the same time.

\halfblankline

\begin{definition}
  The feature $\mu$-calculus~$\muLfprime$ over $\calA$, $\calF$ and~$\calX$,
  consisting of formulas~$\varphi'$, is given by
  \begin{displaymath}
    \varphi' \bnfeq
    \begin{array}[t]{@{}l}
      \falsum \mid \verum \mid \neg \varphi' \mid \varphi' \lor \psi'
      \mid \varphi' \land \psi' \mid 
      \myruby{a|\chi} \mkern1mu \varphi' \mid 
      [a|\chi] \mkern1mu \varphi' \mid 
      X \mid \mu X . \mkern1mu \varphi' \mid 
      \nu X . \mkern1mu \varphi'
    \end{array}
  \end{displaymath}
  where for $\mu X .\mkern1mu \varphi'$ and $\nu X . \mkern1mu
  \varphi'$ all free occurrences of~$X$ in~$\varphi'$ are in the scope
  of an even number of negations.
\end{definition}

\halfblankline

\noindent
The semantics of~$\muLfprime$ is given in terms of product families,
i.e.\ sets of products. Therefore, we consider now sets of
state-family pairs $\lbrace (s_i,P_i) \mid i \in I \rc$, for some
index set~$I$, i.e.\ pairs $(s_i,P_i)$ of states and sets of products, 
rather than sets of state-product pairs $\lbrace (s_i,p_i) \mid i \in I \rc$. 
We put $\sPSet = \powerset{\mkern1mu S \times \powersetP}$.

\halfblankline

\begin{definition}
  \label{df-zeta-semantics}
  Let an FTS~$F$, with set of states~$S$, be given. Define $\sPEnv$,
  the set of state-family envi\-ron\-ments, by $\sPEnv = \calX \to
  \sPSet$. Then the semantic function $\bbFF{\cdot}\!: \muLfprime \to
  \sPEnv \to \sPSet$ is\,given\,by
  \begin{align*}
      \bbFF{ \falsum }(\zeta) & = \varnothing \\
      \bbFF{ \verum }(\zeta) & = S \times \powersetP  \\ 
      \bbFF{ \neg \varphi' }(\zeta) & =
      ( S \times \powersetP )
      \setminus \bbFF{ \varphi' }(\zeta) 
      \displaybreak[2] \\ 
      \bbFF{ ( \varphi' \lor \psi' ) }(\zeta) & =
      \bbFF{ \varphi' }(\zeta)
      \cup
      \bbFF{ \psi' }(\zeta)
      \\
      \bbFF{ ( \varphi' \land \psi' ) }(\zeta) & =
      \bbFF{ \varphi' }(\zeta)
      \cap
      \bbFF{ \psi' }(\zeta)
      \smallskip \displaybreak[3] \\
      \bbFF{ \myruby{a|\chi} \mkern1mu \varphi' }(\zeta) & =
      \lc (s,P) \mid 
      {P \subseteq \chi} \land
      \bigl(
      \exists \mkern1mu \gamma, t \colon s
      \transF{a}{\gamma} t \land 
      {} \\ & \qquad\ \quad\ \quad
      P \subseteq \gamma \land
      (t, P \cap \chi \cap \gamma) \in \bbFF{\varphi'}(\zeta) 
      \bigr) \rc 
      \\
      \bbFF{ \mybox{a|\chi} \mkern1mu \varphi' }(\zeta) & =
      \lc (s,P) \mid 
      P \cap \chi \neq \varnothing \Rightarrow
      \bigl(
      \forall \mkern1mu \gamma, t \colon s \transF{a}{\gamma} t 
      \land 
      {} \\ & \qquad\ \quad\ \quad
      {P \cap \chi \cap \gamma} \neq \varnothing \Rightarrow  
      (t, P \cap \chi \cap \gamma) \in \bbFF{ \varphi'}(\zeta) 
      \bigr) \rc
      \displaybreak[3] \\
      \bbFF{ X }(\zeta) & =
      \zeta(X) 
      \\
      \bbFF{ \mu X . \mkern1mu \varphi' }(\zeta) & =
      \lfp( \, W \mapsto \bbFF{ \varphi' }(\zeta[W/X]) \, )
      \\
      \bbFF{ \nu X . \mkern1mu \varphi' }(\zeta) & =
      \gfp( \, W \mapsto \bbFF{ \varphi' }(\zeta[W/X]) \, )
  \end{align*}
\end{definition}

\halfblankline

\noindent
Well-definedness of the semantic function $\bbFF{ \cdot }$ can be
checked as usual.

For an FTS~$F$ with initial state~$\sinit$, a set of products $P
\subseteq \calP$ and a closed $\muLfprime$-formula~$\varphi'$, 
we write $P \modelsFprime \varphi'$ iff $(\sinit \mkern1mu ,P) \in
\bbF{\varphi'}(\zeta_0)$, where the environment $\zeta_0 \in \sPEnv$
is such that $\zeta_0(X) = \varnothing$ for all $X \in \calX$.

The main difference between $\muLf$ and~$\muLfprime$ lies in their
respective \lq{diamond}\rq\ modalities, $\mydiamond{a|\chi}$ for~$\muLf$
and $\myruby{a|\chi}$ for~$\muLfprime$. Here, for $\myruby{a|\chi}
\mkern1mu \varphi'$ of~$\muLfprime$ and $(s,P) \in \sPSet$ to hold, we
require all products in~$P$ to fulfill the feature expression~$\chi$,
and we require that the witnessing transition $s \transF{a}{\gamma}
t$, guarded by the feature expression~$\gamma$, is the same for all
products in~$P$.

\halfblankline

\begin{example}\label{counter-ex-duality} 
  Consider the FTS given on the right. 
\parpic[r]{%
  \begin{tikzpicture}[>=stealth', every state/.style={draw, minimum 
      size=10pt, inner sep=1pt},node distance=35pt] 
    \node (sx) {};
    \node[state, above=of sx, yshift=-6pt] (s0) {\scriptsize $s_0$};
    \node[state, left=of sx] (s1) {\scriptsize $s_1$};
    \node[state, right=of sx] (s2) {\scriptsize $s_2$};
    \draw[->] 
    (s0)++(0,+0.5) -- (s0)
    (s0) edge node[above left, xshift=5pt, yshift=-2pt] {\scriptsize
      $a|\ttf$} (s1)  
    (s0) edge node[above right, xshift=-5pt, yshift=-2pt] {\scriptsize
      $a|\neg \ttf$} (s2)  
    ;
  \end{tikzpicture} 
}
  The transition from $s_0$ to~$s_1$ 
  is possible for products having the feature~$\ttf$ and 
  the transition from $s_0$ to~$s_2$ for products not having this
  feature. Consider the products $p_1 = \singleton{ \ttf ,\, \ttg \, }$
  and $p_2 = \singleton{ \ttg \, }$. Thus, $p_2$~does not have
  feature~$\ttf$. Then $p_1 \modelsF \mydiamond{a|\verum} \verum$
  because of the transition to~$s_1$, and $p_2 \modelsF
  \mydiamond{a|\verum} \verum$ because of the transition
  to~$s_2$. However, we do \emph{not} have $\singleton{ p_1 ,\, p_2 }
  \modelsFprime \myruby{a|\verum} \verum$, since there is no
  $a$-transition from~$s_0$ possible for both $p_1$ and~$p_2$.
  Note that we do not have $\singleton{ p_1 ,\, p_2 } \modelsFprime
  \mybox{a|\verum} \falsum$ either, showing that the modalities
  $\myruby{a|\chi}$ and~$\mybox{a|\chi}$ of~$\muLfprime$ are not each
  others dual.
\end{example}

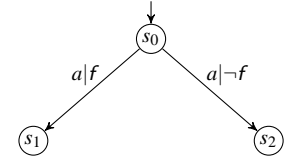
\halfblankline

\noindent
Despite this example, the semantics of
Definition~\ref{df-zeta-semantics} based on sets of state-family pairs
is in several ways consistent with the semantics of
Definition~\ref{df-eta-semantics} based on sets of state-product
pairs, as we will see next.

A formula $\varphi' \in \muLfprime$ has a corresponding formula
$\varphif \in \muLf$: the formulas $\varphi'$ and~$\varphif$ are the
same except that one involves a modality $\myruby{a|\chi}$ where
the other uses~$\mydiamond{a|\chi}$. More specifically, we define a
translation function $\fmf : \muLfprime \to \muLf$ by
\begin{displaymath}
  \def\arraystretch{1.2}
  \begin{array}{r@{\,}c@{\,}lcr@{\,}c@{\,}lcr@{\,}c@{\,}l}
    \fmf( \falsum ) & = & \falsum 
    &&
    \fmf( \varphi' \land \psi' ) & = & 
    \fmf( \varphi') \land \fmf(\psi') 
    &&
    \fmf( X ) & = & X \\
    \fmf( \verum ) & = & \verum &&
    \fmf( \myruby{a|\chi} \mkern1mu \varphi' )
    & = &
    \mydiamond{a|\chi} \mkern1mu \fmf( \varphi' ) &&
    \fmf( \mu X . \mkern1mu \varphi' ) & = & \mu X . \mkern1mu \fmf(
    \varphi' ) \\ 
    \fmf( \varphi' \lor \psi' ) & = & 
    \fmf( \varphi') \lor \fmf(\psi') &&
    \fmf( \mybox{a|\chi} \mkern1mu \varphi' )
    & = &
    \mybox{a|\chi} \mkern1mu \fmf( \varphi' ) &&
    \fmf( \nu X . \mkern1mu \varphi' ) & = & \nu X . \mkern1mu \fmf(
    \varphi' ) \\  
  \end{array}
  \def\arraystretch{1.0}
\end{displaymath}
Moreover, for an environment~$\zeta \in \sPEnv$ and an
environment~$\eta \in \spEnv$ we say that $\zeta$~relates to~$\eta$
if ${(s,\singleton{p}) \in \zeta(X)} \iff {(s,p) \in \eta(X)}$, for
all variables~$X \in \calX$.

\halfblankline

\begin{lemma}
  \label{lemma-zeta-singleton-vs-eta}
  Let an FTS~$F$, with set of states~$S$, a state $s \in S$ and a product
  $p \in \calP$, be given. Suppose the environments $\zeta \in \sPEnv$
  and $\eta \in \spEnv$ are such that $\zeta$~relates to~$\eta$. Then
  it holds that
  \begin{displaymath}
    (s,\singleton{p}) \in \bbFF{\varphi'}(\zeta)
    \quad \text{iff} \quad
    (s,p) \in \bbF{ \fmf( \varphi' )}(\eta)
  \end{displaymath}
  for all $\varphi' \in \muLfprime$.
\end{lemma}

\begin{proof}
  The proof is by induction on the structure of~$\varphi'$. We only
  treat two cases.

  Case~1, $\myruby{a|\chi}$: Recall $\fmf( \myruby{a|\chi} \mkern1mu
  \varphi' ) = \mydiamond{a|\chi} \mkern1mu \varphif$ if $\fmf(
  \varphi' ) = \varphif$. We have
  \begin{displaymath}
    \def\arraystretch{1.2}
    \begin{array}{rcl}
      \multicolumn{3}{l}{(s,p) \in \bbFF{ \myruby{a|\chi} \mkern1mu
          \varphi' }(\zeta)}
      \\ & \iff &
      {\singleton{p} \subseteq \chi} \land
      \bigl(
      \exists \mkern1mu \gamma, t \colon 
      s \transF{a}{\gamma} t \land
      {\singleton{p} \subseteq \gamma} \land
      (t, \singleton{p} \cap \chi \cap \gamma) \in \bbFF{ \varphi'
      }(\zeta) 
      \bigr)
      \\ && \qquad
      \text{(by definition of $\bbFF{ \cdot }$)}
      \\ & \iff &
      {p \in \chi} \land
      \bigl(
      \exists \mkern1mu \gamma, t \colon 
      s \transF{a}{\gamma} t \land
      {p \in \gamma} \land
      (t,p) \in \bbF{ \varphif
      }(\eta) 
      \bigr)
      \\ && \qquad
      \text{($\singleton{p} \cap \chi \cap \gamma = \singleton{p}$ and
        by induction hypothesis)}
      \\ & \iff &
      (s,p) \in \bbF{ \mydiamond{a|\chi} \mkern1mu \varphif }(\eta)
      \\ && \qquad
      \text{(by definition of $\bbF{ \cdot }$)}
    \end{array}
    \def\arraystretch{1.0}
  \end{displaymath}
  Case~2, $\mu X . \mkern1mu \varphi'$: We have $\fmf( \mu X
  . \mkern1mu \varphi' ) = \mu X . \mkern1mu \varphif$ if $\fmf(
  \varphi' ) = \varphif$. Put $W_0 = \varnothing$, $W_{i{+}1} = \bbFF{
    \varphi ' }(\zeta[ W_i/X ] )$ and $V_0 = \varnothing$, $V_{i{+}1}
  = \bbF{ \varphif }(\eta[ V_i/X ])$. We claim
  \begin{equation}
    \label{eqn-claim-zeta-singleton-vs-eta}
    {(s,\singleton{p}) \in W_i} 
    \iff
    {(s,p) \in V_i}
    \quad \text{and} \quad
    \text{$\zeta[ W_i/X ]$ relates to $\eta[ V_i/X ]$}
  \end{equation}
  for all~$i \geqslant 0$.

  \emph{Proof of the claim.} Induction on~$i$. Basis, $i=0$:
  Straightforward. Induction step, $i{+}1$: Note, $\zeta[ W_i/X ]$
  relates to~$\eta[ V_i/X ]$ by induction hypothesis
  for~$i$. Therefore we have $(s,\singleton{p}) \in \bbFF{ \varphi'
  }(\zeta[ W_i/X ]) = W_{i{+}1}$ iff $(s,p) \in \bbF{ \varphif }(\eta[
  V_i/X ]) = V_{i{+}1}$ by induction hypothesis
  for~$\varphi'$. Because of this it follows that $\zeta[ W_{i{+}1}/X
  ]$~relates to~$\eta[ V_{i{+}1}/X ]$.

  Now, unfolding the various definitions, we have $\bbFF{ \mu X
    . \mkern1mu \varphi' }(\zeta) = \bigcup_{i=0}^\infty \: \bbFF{
    \varphi' }(\zeta[ W_i/X ]) = \bigcup_{i=0}^\infty \, W_{i}$ and
  $\bbF{ \mu X . \mkern1mu \varphif }(\eta) = \bigcup_{i=0}^\infty \:
  \bbF{ \varphif }(\eta[V_i/X ]) = \bigcup_{i=0}^\infty \, V_{i}$.
  Using Equation~(\ref{eqn-claim-zeta-singleton-vs-eta}) we obtain
  that $(s,\singleton{p}) \in \bbFF{ \mu X . \mkern1mu \varphi'
  }(\zeta)$ iff $(s,p) \in \bbF{ \mu X . \mkern1mu \varphif }(\eta)$,
  as was to be shown.
\end{proof}


\noindent 
Rephrasing Lemma~\ref{lemma-zeta-singleton-vs-eta} in terms of
satisfaction, we have $\singleton{p} \modelsFprime \varphi'$ iff $p
\modelsF \fmf( \varphi' )$ for $\varphi' \in \muLfprime$ closed.


To relate the semantics of $\muLfprime$ and~$\muLf$, we use a
projection function $\fp : \sPSet \to \spSet$ from sets of
state-family pairs to sets of state-product pairs, given by
$\fp(\lc (s_i,P_i) \mid i \in I \rc = \lc (s_i,p) \mid i \in I ,\, p
\in P_i \rc$.

\begin{theorem}
  \label{theorem-zeta-semantics-vs-eta-semantics}
  Let an FTS~$F$, with state~$s \in S$ and a set of products~$P
  \subseteq \calP$, be given. Suppose $\zeta$ and~$\eta$ are such that
  $\fp(\zeta(X))\subseteq \eta(X)$, for all $X \in \calX$. Then it
  holds that
  \begin{displaymath}
    (s,P) \in \bbFF{\varphi'}(\zeta) 
    \implies 
    \forall p \in P \colon
    (s,p) \in \bbF{ \fmf( \varphi' )}(\eta)
  \end{displaymath}
  for all negation-free $\varphi' \in \muLfprime$.
\end{theorem}

\begin{proof}
  Induction on the structure of~$\varphi'$. We only exhibit two cases.

\vspace*{-0.25\baselineskip}
  Case~1, $\myruby{a|\chi}$: If $(s,P) \in \bbFF{ \myruby{a|\chi}
    \mkern1mu \varphi' }(\zeta)$, then we have $P \subseteq \chi$ and
  $\gamma, t$ exist such that $s \transF{a}{\gamma} t$, $P
  \subseteq \gamma$ and $(t,P) \in \bbFF{ \varphi' }(\zeta)$. Thus,
  since $p \in P$ we have $p \in \chi$, $p \in \gamma$ and $(t,p) \in
  \bbF{ \varphif }(\eta)$ by induction hypothesis. 
  Therefore, $(s,p) \in \bbF{
    \mydiamond{a|\chi} \mkern1mu \varphi }(\eta)$.

  Case~2, $\mu X . \mkern1mu \varphi'$: We have $\bbFF{ \mu X
    . \mkern1mu \varphi' }(\zeta) = \bigcup_{i=0}^\infty \, W_i$ where
  $W_0 = \varnothing$, $W_{i{+}1} = \bbFF{ \varphi' }(\zeta [W_i/X] )$
  and $\bbF{ \mu X . \mkern1mu \varphif }(\eta) = \bigcup_{i=0}^\infty
  \, V_i$ where $V_0 = \varnothing$, $V_{i{+}1} = \bbF{ \varphif
  }(\eta [V_i/X] )$. We claim:
  \begin{equation}
    {\fp(W_i)}
    \subseteq
    {V_i}
    \qquad
    \text{for $i \geqslant 0$}
    \label{eq-Wi-subset-Vi}
  \end{equation}
  \emph{Proof of the claim.} 
  Induction on~$i$. Basis, $i=0$: Clear. Induction step, $i{+}1$: By induction
  hypothesis $\fp(W_i) \subseteq V_i$. Hence, $\fp( \zeta[W_i/X](Y) )
  \subseteq \eta[V_i/X](Y)$ for all $Y \in \calX$. Thus, by the induction
  hypothesis for $\varphi'$, we have that $(s,P) \in \bbFF{ \varphi'
  }(\zeta [W_i/X] )$ implies $(s,p) \in \bbF{ \varphif }(\eta [V_i/X]
  )$ given $p \in P$, hence $\fp( \bbFF{ \varphi' }(\zeta [W_i/X] )
  \subseteq \bbF{ \varphif }(\eta [V_i/X] )$. Therefore,
  $\fp(W_{i{+}1}) \subseteq V_{i{+}1}$, as was to be shown.

  Now, from claim~(\ref{eq-Wi-subset-Vi}), exploiting the continuity
  of~$\fp$, we obtain $\fp( \, \bbFF{ \mu X . \mkern1mu \varphi'
  }(\zeta) \, ) = \fp( \bigcup_{i=0}^\infty \, W_i ) =
  \bigcup_{i=0}^\infty \, \fp( W_i ) \subseteq \bigcup_{i=0}^\infty \,
  \fp( V_i ) = \fp( \, \bbF{ \mu X . \mkern1mu \varphif }(\eta) \,
  )$, or phrased differently $(s,P) \in \bbFF{ \mu X . \mkern1mu
    \varphi' }(\zeta)$ implies $(s,p) \in \bbF{ \mu X . \mkern1mu
    \varphif }(\eta) \, )$ which completes the induction step for $\mu
  X . \mkern1mu \varphi'$.
\end{proof}

\halfblankline

\noindent
In terms of satisfaction,
Theorem~\ref{theorem-zeta-semantics-vs-eta-semantics} can be
reformulated as 
\begin{equation}
  \label{eqn-zeta-models-vs-eta-models}
  P \modelsFprime \varphi' 
  \implies 
  \forall \mkern1mu p \in P \colon 
  p \modelsF \fmf(\varphi')
\end{equation}
for all $P \subseteq \calP$, $\varphi' \in \muLfprime$ closed. Thus,
given a set of products and a formula~$\varphif \in \muLf$, instead of
verifying~$\varphif$ for each individual product, we may seek to
verify the corresponding $\varphi'$ into~$\muLfprime$ of~$\varphi$ at
once for the complete set of products. In case of an affirmative
answer, Equation~(\ref{eqn-zeta-models-vs-eta-models}) guarantees that
the formula~$\varphif$ will hold for the separate products.

For formulas in the intersection of $\muLf$ and~$\muLfprime$ that are
negation-free, i.e.\ formulas without any negation or modalities
$\mydiamond{a|\chi}$ or~$\myruby{a|\chi}$, we have a stronger result.

\halfblankline

\begin{theorem}
  \label{lemma-zeta-semantics-vs-eta-semantics-without-diamonds}
  Suppose the formula $\varphi' \in\muLf \cap \muLfprime$ is
  negation-free, and $\zeta \in \sPEnv$ and $\eta \in \spEnv$ are
  such that $(s,P) \in \zeta(X) \iff \forall p \in P \colon (s,p) \in
  \eta(X)$, for all $s \in S$, $P \subseteq \calP$ and $X \in
  \calX$. Then it holds that $$(s,P) \in \bbFF{\varphi'}(\zeta) \iff
  \forall \mkern1mu p \in P \colon (s,p) \in \bbF{\varphi'}(\eta)$$
  for all states $s \in S$ and sets of products $P \subseteq \calP$.
\end{theorem}

\begin{proof}
  Induction on the structure of~$\varphi$. We exhibit two
  cases only. 

  Case~1, $\mybox{a|\chi}$: ($\Longrightarrow$) Suppose $(s,P) \in \bbFF{
    \mybox{a|\chi} \varphi }(\zeta)$ and pick $p \in P$. As in the
  proof of Theorem~\ref{theorem-zeta-semantics-vs-eta-semantics}, we
  reason as follows: Suppose $p \in \chi$ and $\gamma, t$ are such
  that $s \transF{a}{\gamma} t$ and $p \in \gamma$. Then we have $P
  \cap \chi \neq \varnothing$ and $P \cap \chi \cap \gamma \neq
  \varnothing$. Hence $(t,P \cap \chi \cap \gamma) \in \bbFF{ \varphi'
  }(\zeta)$ by definition of $\bbFF{ \mybox{a|\chi} \varphi'
  }(\zeta)$. Since $p \in P \cap \chi \cap \gamma$ it follows that
  $(t,p) \in \bbF{ \varphi }(\eta)$ by induction hypothesis.

  ($\Longleftarrow$) Suppose $(s,p) \in \bbF{ \mybox{a|\chi} \varphi
  }(\eta)$ for all $p \in P$. Assume furthermore, $P \cap \chi \neq
  \varnothing$ and $\gamma, t$ are such that $s \transF{a}{\gamma} t$
  and $P \cap \chi \cap \gamma \neq \varnothing$. Clearly, for all $p
  \in P \cap \chi \cap \gamma$ we have $p \in \chi$ and $p \in
  \gamma$. Thus, by definition of $\bbF{ \mybox{a|\chi} \varphi
  }(\eta)$, we have $(t,p) \in \bbF{ \varphi }(\eta)$ for all $p \in P
  \cap \chi \cap \gamma$. By induction hypothesis, $(t,P \cap \chi
  \cap \gamma) \in \bbFF{ \varphi }(\zeta)$. It follows that $(s,P)
  \in \bbFF{ \mybox{a|\chi} \varphi }(\zeta)$ by definition of $\bbFF{
    \mybox{a|\chi} \varphi }(\zeta)$.

  Case~2, $\mu X . \mkern1mu \varphi$: We have that $\bbFF{ \mu X
    . \mkern1mu \varphi }(\zeta) = \bigcup_{i=0}^\infty \, Z_i$ where
  $Z_0 = \varnothing$, $Z_{i{+}1} = \bbFF{ \varphi }(\zeta [Z_i/X] )$
  and $\bbF{ \mu X . \mkern1mu \varphi }(\eta) = \bigcup_{i=0}^\infty
  \, V_i$ where $V_0 = \varnothing$, $V_{i{+}1} = \bbF{ \varphi }(\eta
  [V_i/X] )$. We claim:
  \begin{equation}
    (s,P) \in Z_i \iff \forall p \in P \colon (s,p) \in V_i
    \qquad
    \text{for $i \geqslant 0$}
    \label{eq-Zi-subset-Vi-no-diamonds}
  \end{equation}
  \emph{Proof of the claim.}
  Induction on~$i$. Basis, $i=0$: Trivial. Induction step, $i{+}1$: By
  induction hypothesis and the assumption for $\zeta$ and~$\eta$ we
  have, for any $s \in S$, any $P \subseteq \calP$ and
  $Y \in \calX$, that $(s,P) \in \zeta [Z_i/X](Y)$ iff $(s,p) \in \eta
  [V_i/X](Y)$ for all $p \in P$. Thus, $(s,P) \in Z_{i{+}1}$ iff
  $(s,P) \in \bbFF{ \varphi }(\zeta [Z_i/X])$ iff $\forall p \in P$:
  $(s,p) \in \bbF{ \varphi }(\eta [V_i/X])$, by induction hypothesis
  for~$\varphi$, iff $\forall p \in P$: $(s,p) \in V_{i{+}1}$.

  Now, using claim~(\ref{eq-Zi-subset-Vi-no-diamonds}), we derive,
  for $s \in S$ and $P \subseteq \calP$, $(s,P) \in
  \bigcup_{i=0}^\infty \: Z_i$ iff $\exists \mkern1mu i \geqslant 0$:
  $(s,P) \in Z_i$ iff $\exists \mkern1mu i \geqslant 0$: $\forall p
  \in P$: $(s,p) \in V_i$ iff $\forall p \in P$: $(s,p) \in
  \bigcup_{i=0}^\infty \: V_i$. For the latter equivalence we use that
  $( V_i )_{i=0}^\infty$ is an ascending chain, i.e.\ $V_i \subseteq
  V_{i{+}1}$ for all~$i$, and that the set~$P$ is finite. From this it
  follows that $(s,P) \in \bbFF{ \mu X . \varphi }(\zeta)$ iff
  $\forall p \in P$: $(s,p) \in \bbF{ \mu X . \varphi }(\eta)$, for
  all $s \in S$ and $P \subseteq \calP$.
\end{proof}

\halfblankline

\noindent
Following the above theorem, the strengthening of
Equation~(\ref{eqn-zeta-models-vs-eta-models}) for closed, negation-free as
well as $\mydiamond{a|\chi}$-free and~$\myruby{a|\chi}$-free feature
$\mu$-formulas reads
\begin{gather}
  \label{eqn-zeta-models-vs-eta-models-special-case}
  P \modelsFprime \varphi' 
  \iff
  \forall \mkern1mu p \in P \colon 
  p \modelsF \varphi'
\end{gather}
for all $P \subseteq \calP$, $\varphi' \in \muLfprime \cap \muLf$
closed and negation-free.


\section{Family-based model checking of~$\muLfprime$}
\label{sec-mc-zeta}

In view of Theorem~\ref{theorem-zeta-semantics-vs-eta-semantics} and
Equation~(\ref{eqn-zeta-models-vs-eta-models}), as presented in the
previous section, we may divert to verifying $P \modelsFprime
\varphi'$ family-wise when we aim to check a property~$\varphi$
expressed in the feature $\mu$-calculus~$\muLf$ for a set of
products~$P$. The $\muLfprime$-formula~$\varphi'$ is obtained from the
$\muLf$-formula~$\varphi$ just by replacing
modalities~$\mydiamond{a|\chi}$ by modalities~$\myruby{a|\chi}$. However,
the granularity for model checking for~$\muLfprime$ is that of sets of
products rather than individual products as for~$\muLf$.

In principle, a dedicated model checker for~$\muLfprime$ can be built
starting from Definition~\ref{df-zeta-semantics} following a recursion
scheme. However, in such a scenario specific optimization techniques
and performance enhancing facilities need to be constructed from
scratch. Here, we sketch an alternative approach. The problem of
deciding $P \modelsFprime \varphi'$ can be formalized using
multi-sorted first-order modal $\mu$-calculus, hereafter referred to
as $\muLFO$, as proposed by~\cite{GM98,GW05}. Such a calculus is given
in the context of a data signature $\Sigma = (S,O)$, where $S$~is a
set of sorts and $O$~is a set of operations, and of a set of sorted
actions~$\calA$.  Formulas~$\varphi$ and sorted actions~$a(\varv)$ of
(a fragment of)~$\muLFO$ are given by
\begin{displaymath}
    \begin{array}{l}
    \varphi \bnfeq 
    \begin{array}[t]{@{}l}
      b \mid \neg \varphi \mid \varphi \lor \psi
      \mid \varphi \land \psi 
      \mid \exists \varv \ofsort \sortD. \varphi 
      \mid \forall \varv \ofsort \sortD. \varphi \mid 
      {} \smallskip \\
      \mydiamond{a(\varv)} \mkern1mu \varphi \mid 
      \mybox{a(\varv)} \mkern1mu \varphi \mid 
      X(t)  \mid 
      \mu X \mkern-1mu (\varv\ofsort \sortD = t) . \varphi 
      \mid \nu X \mkern-1mu (\varv\ofsort \sortD = t) . \varphi\\
    \end{array}
    \end{array}
  \end{displaymath}
where for $\mu X \mkern-1mu (\varv \ofsort \sortD = t) . \varphi$ and
$\nu X \mkern-1mu (\varv \ofsort \sortD = t). \varphi$ all free
occurrences of~$X$ in~$\varphi$ are in the scope of an even number of
negations, $b$~is an expression of Boolean sort, $t$ is an arbitrary
expression and $\varv$~is a variable of sort $\sortD$.

Expressions over $\muLFO$ are to be interpreted over specific LTS
where actions carry parameters. The expressions given by~$\alpha$ in
the above grammar allow for reasoning about sets of such parametrized
actions. In our approach to SPL, we use these parameters to model sets
of products. Moreover, we use the parameter~$\varv$ of the fixpoint
operators to keep track of the set of products for which we are
evaluating the fixpoint formula. Note, this set is dynamic: whenever
we encounter a modality such as $\mybox{a|\chi} \mkern1mu \varphi$ or
$\myruby{a|\chi} \mkern1mu \varphi$, the set of products for which we
need to evaluate $\varphi$ is restricted by $\chi$. The bottom line
is that we can devise a translation~$T$ that answers whether $(s,P) \in
\bbFF{\varphi}$, where $s$ is a state of an FTS, by answering $\bar{s}
\in \bbB{T(P,\varphi)}$, where $\bar{s}$ is a state of an LTS with
parametrized actions. While a detailed exposition is beyond the scope
of the current paper, we illustrate the approach using a small SPL
example.

\halfblankline

\begin{example}
  \label{ex-coffee-machine-cont}
  Consider the FTS $F$ modeling a family of (four) coffee machines
  from Example~\ref{ex-coffee-machine}, recalled below (left).  Also
  depicted below (right) is the LTS ${L}(F)$ with parametrized actions
  that represents $F$.
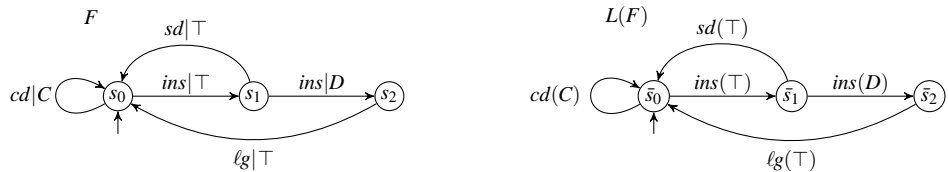
\begin{center}
  \begin{tikzpicture}[>=stealth',every state/.style={draw, minimum
      size=10pt,inner sep=1pt},node distance=40pt] 
    \node[state] (s0) {\scriptsize $s_0$};
    \node[state, right=of s0] (s1) {\scriptsize $s_1$};
    \node[state, right=of s1] (s2) {\scriptsize $s_2$};
    \path ([xshift=-10pt,yshift=30pt] s0) node {\scriptsize $F$}; 

    \draw[->] 
    (s0)++(0,-0.5) -- (s0)
    (s0) edge node[yshift=-3pt,above] {\scriptsize $\mathit{ins} |
      \mkern1mu \verum$} (s1) 
    (s1) edge[out=105,in=75] node[yshift=-2pt,above] {\scriptsize
      $\mathit{sd} | \mkern1mu \verum$} (s0)
    (s1) edge node[yshift=-3pt,above] {\scriptsize $\mathit{ins} |
      \mkern0mu D$} (s2) 
    (s2) edge[bend left] node[yshift=0pt,below] {\scriptsize ${\ell}
      \mkern-1.5mu {g} | \mkern1mu \verum$} (s0)
    (s0) edge[loop,min distance=10mm,out=210,in=150]
    node[left,xshift=2pt] {\scriptsize $\mathit{cd} | \mkern1mu C$} (s0);

    \node[state, right of=s2,xshift=60pt] (t0) {\scriptsize $\bar{s}_0$};
    \node[state, right=of t0] (t1) {\scriptsize $\bar{s}_1$};
    \node[state,right =of t1] (t2) {\scriptsize $\bar{s}_2$};
    \path ([xshift=-10pt,yshift=30pt] t0) node {\scriptsize ${L}(F)$}; 

    \draw[->] 
    (t0)++(0,-0.5) -- (t0)
    (t0) edge node[yshift=-3pt,above] {\scriptsize $\mathit{ins}(\verum)$} (t1)
    (t1) edge[out=105,in=75] node[yshift=-2pt,above] {\scriptsize
      $\mathit{sd}(\verum)$} (t0) 
    (t1) edge node[yshift=-3pt,above] {\scriptsize $\mathit{ins}(D)$} (t2)
    (t2) edge[bend left] node[yshift=0pt,below] {\scriptsize $\ell
      \mkern-1.5mu {g} (\verum)$} (t0) 
    (t0) edge[loop,min distance=10mm,out=210,in=150] node[left]
    {\scriptsize $\mathit{cd}(C)$} (t0);

\end{tikzpicture}
\end{center}
Let $\varphi'$ be the $\muLfprime$-formula $\nu X. \mu Y. \bigl( \, (
\mybox{ \mathit{ins} | E } Y \land \mybox{ \mathit{cd}| E }Y \land
\mybox{ \ell \mkern-1.5mu {g} | E} Y ) \land \mybox{\mathit{sd} | E} X
\bigr)$ expressing that on all infinite runs involving actions
$\mathit{ins}$, $\mathit{cd}$, ${\ell \mkern-1.5mu {g}}$,
and~$\mathit{sd}$, the action~$\mathit{sd}$ occurs infinitely
often. Note that this formula holds in state~$s_0$ and only for
products that do not feature the dollar unit (feature~$D$) nor the
cleaning and descale unit (feature~$C$). For, if the cleaning and
descale unit is present, there is a violating infinite run consisting
of $\mathit{cd}$-actions only, whereas when the dollar unit is
present, there is an infinite run containing only $\mathit{ins}$ and
$\ell \mkern-1.5mu {g}$ actions.

Assume that the sort $\PSet$ represents the set $\bftwo^{\calF}$. The
set of features~$\calF$ is finite and, therefore, the sort~$\PSet$ is
easily defined and can be used to effectively compute with. Moreover,
we presume all usual set operators on~$\bftwo^{\calF}$ to have
counterparts for $\PSet$ too. The $\muLFO$-formula $T(P,\varphi')$ that
corresponds to $\varphi' \in \muLfprime$ above is then the following:
\begin{displaymath}
  \begin{array}{l}
    \nu X(P_x \ofsort \PSet = P).  \mu Y(P_y \ofsort \PSet = P_x). (
    {} \smallskip \\ \qquad
    (P_y \cap E = \varnothing \lor \forall e \ofsort
    \PSet. \mybox{ \,\! \mathit{ins}(e) \,\! }(P_y \cap E \cap e = \varnothing
    \lor Y(P_y \cap E \cap e) )) 
    \land {} \\ \qquad
    (P_y \cap E = \varnothing \lor \forall e \ofsort
    \PSet. \mybox{\,\, \mathit{cd}(e) \,\! }(P_y \cap E \cap e = \varnothing \lor
    Y(P_y \cap E \cap e) )) 
    \land {} \\ \qquad
    (P_y \cap E = \varnothing \lor \forall e \ofsort
    \PSet. \mybox{ \,\, {\ell \mkern-1.5mu {g} \, }(e)\,\!}(P_y \cap E \cap e =
    \varnothing \lor Y(P_y \cap E \cap e) ))  
    \land {} \\ \qquad
    (P_y \cap E = \varnothing \lor \forall e \ofsort
    \PSet. \mybox{\,\, \mathit{sd}(e) \, \! }(P_y \cap E \cap e =
    \varnothing \lor X(P_y \cap E \cap e) )) 
\ )
\end{array}
\end{displaymath}
In the resulting formula~$T(P,\varphi')$, the modal operators
of~$\varphi'$ are essentially mapped to the modal operators
of~$\muLFO$ and the information concerning the feature expressions is
handled by the data parameters and appropriate conditions, mirroring
the semantics of $\muLfprime$. Deciding whether $(s_0 \mkern1mu ,\neg (C
\lor D)) \in \bbFF{\varphi'}$ then translates to verifying whether
$\bar{s}_0 \in \bbB{T(\neg(C \lor D),\varphi')}$. The latter can be
done using a toolset such as \mCRL, which supports~$\muLFO$ and which
allows for representing LTS with parametrized actions.
\end{example}

\section{Concluding remarks and future work}
\label{conclusion}

We have introduced two variants of the modal $\mu$-calculus, each with
explicit FTS semantics, by incorporating feature expressions, and we
have compared these logics among each other. This resembles work done
for LTL~\cite{CCSHLR13} and CTL~\cite{CCHLS14}, but we can also
express typical $\mu$-calculus properties not expressible in LTL or
CTL\@. We have then shown how to achieve family-based model checking
of SPL with the existing \mCRL{} toolset by exploiting an embedding of
the newly introduced feature-oriented $\mu$-calculus variant, with an
FTS semantics in terms of sets of products, into the $\mu$-calculus
with data. 
It follows that from a logical point of view the featured
$\mu$-calculi proposed here are a sublogic of the $\mu$-calculus with
data. However, methodologically the new modalities in~$\muLf$
and~$\muLfprime$ highlight via feature expressions the variability and
allow to direct the analysis to specific families of products, which
may give better insight and, in the preliminary casestudies conducted,
quicker response times of the model checker.
In~\cite{BDVW16b}, we evaluate the application of our
approach to a larger SPL model from the literature, the well-known
minepump model.

In this paper we have considered family-based \emph{model checking},
but the same principle has also been applied to other analysis
techniques, like theorem proving, static analysis, and type
checking~\cite{TAKSS14}.  Oftentimes this requires extending existing
tools, but in some cases---like the one described in this paper---the
analysis problem can be encoded in an existing specification language
to allow off-the-shelf tools to be reused. In~\cite{TSHA12}, for
instance, all Java feature modules and their corresponding
feature-based specifications in a Java Modeling Language extension are
translated into a single family-based meta\-specification that can be
passed as-is to the \texttt{K\MakeLowercase{e}Y} theorem
prover~\cite{BHS07}.

Finally, we outline a couple of possibilities for future work.  First,
the main results presented in this paper
(Theorems~\ref{theorem-eta-semantics-vs-eps-semantics}
and~\ref{theorem-zeta-semantics-vs-eta-semantics}) demonstrate that
the validity of a formula for a family of products implies its
validity for the family's individual products. As is done for MTS
in~\cite{BFGM15b}, it would be interesting to establish complementary
results concerning the preservation of the \emph{in\/}validity of a
formula for a product family by the family's individual products. One
possible way to try to achieve this is by providing an alternative
semantics for the diamond and box operators denoting may and must 
modalities.

Second, the implication of
Theorem~\ref{theorem-zeta-semantics-vs-eta-semantics} is strengthened
to an equivalence in
Theorem~\ref{lemma-zeta-semantics-vs-eta-semantics-without-diamonds}
(i.e.\ a formula is valid for a family of products iff it is valid for
the family's individual products) for a specific subset of feature
$\mu$-calculus formulas. It would be interesting to study the
conditions under which this equivalence can be obtained for a
larger set of feature $\mu$-calculus formulas. Possible strategies
include considering an alternative set of FTS (e.g.\ with a different
structure) or, following~\cite{CCSHLR13,CCHLS14}, separating the feature
expressions from the diamond and box operators and instead
parametrizing each formula with a feature-based operator that
quantifies the specific set of products for which the formula has to
be verified.

\subsection*{Acknowledgments}

Maurice ter Beek is supported by the EU FP7--ICT FET--Proactive project
QUANTICOL, 600708. 

\bibliographystyle{eptcs}
\bibliography{fmsple16}

\end{document}